\documentclass{lmcs}


\ACMCCS{[{\bf Theory of computation}]: Formal languages and automata theory---Transducers.}

\usepackage{graphicx}
\usepackage{amssymb}
\usepackage{hyperref}

\newcommand{\trunc}{{\mathrm{Trunc}}}
\newcommand{\rest}{{\mathrm{Rest}}}
\newcommand{\cw}{{\mathrm{Cw}}}
\newcommand{\remove}{{\mathrm{Remove}}}
\newcommand{\rev}{{\mathrm{Rev}}}
\DeclareMathOperator{\tuplefy}{Tuplefy}
\DeclareMathOperator{\pair}{Pair}
\DeclareMathOperator{\harv}{harv}
\DeclareMathOperator{\unpad}{Unpad}
\DeclareMathOperator{\mask}{Mask}
\newcommand*\colv[1]{\begin{pmatrix}#1\end{pmatrix}}

\newcommand{\ors}{{\mathcal{O}}}

\begin{document}

\title[Generating Regular Functions]{Generating the Functions with Regular Graphs under Composition}

\author[T.~Kern]{Thomas Kern}	
\address{Department of Mathematics, Cornell University, Malott Hall, Ithaca, NY, 14853, USA}	
\email{trk43@cornell.edu}  
\thanks{This material is based upon work supported by the National Science Foundation Graduate Research Fellowship under Grants No. DMS-0852811 and DMS-1161175. This paper is adapted from the first chapter of my dissertation research under Anil Nerode at Cornell University.}	

\begin{abstract}
\noindent While automata theory often concerns itself with regular predicates, relations corresponding to acceptance by a finite state automaton, in this article I study the regular functions, such relations which are also functions in the set-theoretic sense. Here I present a small (but necessarily infinite) collection of (multi-ary) functions which generate the regular functions under composition. To this end, this paper presents an interpretation of the powerset determinization construction in terms of compositions of input-to-run maps. Furthermore, known results using the Krohn-Rhodes theorem to further decompose my generating set are spelled out in detail, alongside some coding tricks for dealing with variable length words. This will include two clear proofs of the Krohn-Rhodes Theorem in modern notation.
\end{abstract}

\maketitle

\section{Introduction}

Automata theory is particularly fruitful in terms of equivalence theorems: regular expressions, deterministic and nondeterministic automata, the Myhill-Nerode theorem, the regular word logic and the weak second order theory of one successor all are equally expressive in the languages they describe. In this paper, I concern myself not with regular predicates (predicates which hold only for the words in a regular language) but with regular functions, functions whose behavior can be recognized by an automaton. This allows a translation of the Krohn-Rhodes theorem into yet another equivalent.

The Krohn-Rhodes Theorem concerns itself with \emph{finite state transducers}, an abstraction of systems that:
\begin{itemize}
\item Accept inputs from a discrete set at discrete times,
\item Retain some memory about previous inputs, which updates whenever an input is read,
\item For each input read, produce some output from a discrete set based on the input and memory.
\end{itemize}
These are an abstraction of synchronous (as opposed to those that update continuously), digital (as opposed to those that deal with analog values) systems.

Originally proved in \cite{krohn}, the Krohn-Rhodes theorem itself gives a decomposition of arbitrary finite state transducers into a cascade of transducers from a small generating set. Computational implementations of this decomposition are available \cite{nagy05}. The Krohn-Rhodes Theorem can be used to analyze the rough behavior of automata, providing applications to Artificial Intelligence \cite{nagy06}.

Finite state transducers are formalized as follows:
\begin{defi}
A \emph{Moore Machine} is a tuple: $(\Sigma, Q, q_0, \Gamma, \delta, \epsilon)$.
\begin{itemize}
\item $\Sigma$ is a finite set of input characters (alphabet).
\item $Q$ is a finite set of states.
\item $q_0 \in Q$ is the initial state.
\item $\Gamma$ is a finite set of output characters.
\item $\delta: Q \times \Sigma \to Q$ is the transition function.
\item $\epsilon: \Sigma \to \Gamma$ is the output function.
\end{itemize}

For convenience, I will sometimes denote the map $x \mapsto \delta(a,x)$ by $\delta_a$.

Given an input word $w \in \Sigma^*$, construct a run $r$ and output $o$ such that:
\begin{itemize}
\item $r[0] = q_0$.
\item $r[i+1] = \delta(r[i],w[i])$ for $0 \leq i < |w|$.
\item $o[i] = \epsilon(r[i])$, for $0 \leq i \leq |w|$.
\end{itemize}
Where the indexing notation is such that: \[w = w[0], \ldots, w[|w|-1].\]
\end{defi}

The correct way of thinking about Moore Machines is that each input acts as a transition between one state and the next, or that each state is the state between inputs. A proper representation would have input characters half a step offset from states. Outputs are simply a product of the state the automaton is in and so should be in step with the states. However, representing the sequences of input characters and states as words requires a choice of direction to shift half a step. A considerable effort has been made to pick the option to result in the cleanest presentation. In this paper, the input character $w[i]$ tells the device how to transition from state $r[i]$ to state $r[i+1]$ \footnote{This is as opposed to the input character $w[i]$ telling the device how to transition from state $r[i-1]$ to state $r[i]$. This alternative is not uncommon.}.

\begin{exa}
This example below shows how inputs, states, and outputs, respectively, line up according to my notation.

\begin{center}\includegraphics{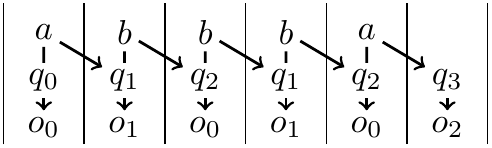}\end{center}

\end{exa}

I will typically consider Moore Machines that simply output their states:
\begin{defi}
A Moore Machine is said to be \emph{transparent} if $\Gamma = Q$ and $\epsilon$ is the identity.
In this case, Moore Machines are presented as a tuple: $(\Sigma, Q, q_0, \delta)$.
\end{defi}

Moore Machines can be interpreted as functions from input words to output words:
\begin{defi}
Given a Moore Machine $M = (\Sigma, Q, q_0, \Gamma, \delta, \epsilon)$, and word $w \in \Sigma^*$, I denote by $M(w)$ the output of $M$ on input $w$, with $M(w) \in \Gamma^*$.
\end{defi}

Moore Machines, however, represent only a small subset of those functions on words which can be reasoned about using finite automata.

\begin{defi}
Given two alphabets, $\Sigma, \Gamma$, a function on words $f: \Sigma^* \to \Gamma^*$ is said to be:
\begin{itemize}
\item \emph{length-preserving} if for any word $w$, $|f(w)| = |w|$. 
\item \emph{causal} if $f(w)[i]$ depends only on $w[0], \ldots, w[i]$.
\item \emph{strictly causal} if $f(w)[i]$ depends only on $w[0], \ldots, w[i-1]$.
\item \emph{character-wise} if $f(w)[i]$ depends only on $w[i]$.
\end{itemize}
\end{defi}

\begin{defi}
Define the following useful functions for dealing with words:
\begin{itemize}
\item Let $S_a$ denote the \emph{successor-}$a$ function which appends the character $a$ to the end of a word. 
\item Let $\trunc$ denote the function which removes the last character of a word.
\item Let $\rest$ denote the function which removes the first character of a word.
\end{itemize}
\end{defi}

\begin{prop}
Given a Moore Machine $M$, the function $M$ computes is strictly causal, and increases length by 1.
\end{prop}

It is often more convenient to deal with length preserving versions of this function:

\begin{defi}
Given a Moore Machine $M = (\Sigma, Q, q_0, \Gamma, \delta, \epsilon)$, and word $w \in \Sigma^*$, denote by:
\begin{itemize}
\item $M^\trunc(w)$ the output of $M$ on input $w$ with the last state removed.
\item $M^\rest(w)$ the output of $M$ on input $w$ with the first state (the start state) removed.
\end{itemize} 
\end{defi}

\begin{prop}
Given a Moore Machine $M$, the function $M^\trunc$ is strictly causal and length preserving, and the function $M^\rest$ is causal and length preserving.
\end{prop}

The notion of finite automaton, or finite state recognizer, is more commonly studied than the finite state transducer. \emph{Finite state automata} are an abstraction of systems that:
\begin{itemize}
\item Accept inputs from a discrete set at discrete times,
\item Retain some memory about previous inputs, which updates whenever an input is read,
\item Having finished reading a sequence of inputs, either \emph{accepts} or \emph{rejects}.
\end{itemize}

Just as finite state transducers can be interpreted as functions on words, finite state automata can be interpreted as predicates on words, returning a boolean value after having read in a word. Their predicative nature means that finite state automata are more convenient to use in applications to formal logic. On the other hand, real world systems are more often interested in transforming inputs, and so are better represented by finite state transducers.

Finite automata yield a notion of a regular set or regular event, a collection of words or sequences of inputs which are exactly those which some finite automaton accepts. Once I define what it means to represent a function of words with an automaton, this will yield a notion of regular function.

Finite state recognizers are formalized as follows:

\begin{defi}
A \emph{finite state automaton} is a tuple: $(\Sigma, Q, I, \delta, F)$.
\begin{itemize}
\item $\Sigma$ is a finite set of input characters (alphabet).
\item $Q$ is a finite set of states.
\item $I \subseteq Q$ is the set of initial states.
\item $\delta: Q \times \Sigma \to Q$ is the transition relation.
\item $F \subseteq Q$ is the set of final states.
\end{itemize}
Given an input word $w \in \Sigma^*$, $r \in Q^*$ is a run on input $w$ if:
\begin{itemize}
\item $r[0] \in I$.
\item $\delta(r[i],w[i],r[i+1])$ for $0 \leq i < |w|$.
\end{itemize}
$w$ is \emph{accepted} if there is a run $r$ on input $w$ such that $r[|w|] \in F$.
\end{defi}

\begin{defi}
A finite state automaton $A = (\Sigma,Q,I,\delta,F)$ is \emph{deterministic} if:
\begin{itemize}
\item $I$ is a singleton.
\item $\delta$ is a function from $Q \times \Sigma \to Q$, that is, given a $q \in Q$, and $a \in \Sigma$, there is a unique $q' \in Q$ such that $\delta(q,a,q')$.
\end{itemize}
By default, $A$ is nondeterministic.
\end{defi}

A well known theorem of finite automata is that:
\begin{prop}\label{detlemma}
If $R$ is the set of accepted inputs of some automaton, then it is also the set of accepted inputs of some deterministic automaton.
\end{prop}
In either case, $R$ is \emph{regular}.

A common convention in logic is to identify a function $f$ with the relation $R_f$ which consists of all pairs of the form $(x,f(x))$, or, for $n$-ary functions, \[(x_0,\ldots,x_{n-1},f(x_0,\ldots,x_{n-1})),\] for $x$ in the domain of $f$. As such, a regular function can be defined as a relation which is regular and also a function. The question now is how to input multiple words, especially multiple words of different lengths, to a finite automaton \footnote{Note that asynchronous input, that is, reading in the input words one at a time separated by a distinguished character is almost completely useless in terms of the functions that can be represented.}

I introduce the $\tuplefy$ map to merge words together in parallel so they can be read by an automaton. For words of different lengths, I add a dummy character $\#$.

\begin{defi}
Define the map \[\tuplefy: \Sigma_0^* \times \cdots \times \Sigma_{n-1}^* \to ((\Sigma_0 \cup \{\#\}) \times \cdots \times (\Sigma_{n-1} \cup \{\#\}))^*,\]
Which satisfies:
\[\tuplefy(w_0,\ldots,w_{n-1})[i]_j = \begin{cases} w_j[i] & \text{ it exists} \\ \# & \text{ otherwise} \end{cases}\]
and $\left|\tuplefy(w_0,\ldots,w_{n-1})\right| = \max_i |w_i|$.
\end{defi}

\begin{exa}
Below is shown how $\tuplefy$ combines words together into one word:

\[
\begin{array}{|r||c|c|c|c|c|}
\hline
w_0 & a & b & b & a &  \\
w_1 & a & b & & &  \\
w_2 &&&&&\\
w_3 & b & b & b & & \\
w_4 & a & a & a & a & a\\
\hline
\tuplefy(w_0,w_1,w_2,w_3,w_4) &\colv{a\\a\\\#\\b\\a} & \colv{b\\b\\\#\\b\\a} & \colv{b\\\#\\\#\\b\\a} & \colv{a\\\#\\\#\\\#\\a} & \colv{\#\\\#\\\#\\\#\\a}\\
\hline
\end{array}
\]
\end{exa}

Now I can define the notion of regular relation and regular function:

\begin{defi}
An $n$-ary relation $R \subseteq \Sigma_0^* \times \cdots \times \Sigma_{n-1}^*$ is \emph{regular} if there is a finite automaton $A$ with input alphabet $(\Sigma_0 \cup \{\#\}) \times \cdots \times (\Sigma_{n-1} \cup \{\#\})$ such that:
\[R(w_0,\ldots,w_{n-1}) \iff A \text{ accepts } \tuplefy(w_0,\ldots,w_{n-1}).\]

An $n$-ary function $f: \Sigma_0^* \times \cdots \times \Sigma_{n-1}^* \to \Sigma_n^*$ is \emph{regular} if there is a finite automaton $A$ with input alphabet $(\Sigma_0 \cup \{\#\}) \times \cdots \times (\Sigma_n \cup \{\#\})$ such that:
\[f(w_0,\ldots,w_{n-1}) = w_n \iff A \text{ accepts } \tuplefy(w_0,\ldots,w_n).\]
\end{defi}

Regular relations and functions are a key part of the analysis of various automaton logics. For instance:
\begin{defi}
Given a finite alphabet $\Sigma$, let $\mathcal{W}_{\Sigma} = (\Sigma^*, \leq, =_{el}, S_a |_{a \in \Sigma})$ where:
\begin{itemize}
\item $\leq$ is the prefix relation on words,
\item $=_{el}$ is the equal length relation on words,
\item $S_a$ is the \emph{Successor-$a$} unary operation, which appends an $a$ onto the end of a word.
\end{itemize}
I call $\mathcal{W}_\Sigma$ the \emph{regular word logic over $\Sigma$} \footnote{While this language may at first seem artificial, it is equally expressive with the \emph{Weak Second Order Theory of One Successor}, the theory of natural numbers, finite sets of natural numbers, the $+1$ operation, and containment.}.
\end{defi}

\begin{prop}\label{wordlogic}
If $R$ is a relation on $\Sigma^*$, the following are equivalent:
\begin{itemize}
\item $R$ is regular,
\item $R$ is given by a formula $\phi$ in the language of $\mathcal{W}_{\Sigma}$.
\end{itemize}
\end{prop}

Hence my interest in regular functions. If $\phi(x_0,\ldots,x_n)$ is a formula in the language of $\mathcal{W}_\Sigma$ such that: \[\forall x_0,\ldots,x_{n-1} \exists x_n: \phi(x_0,\ldots,x_n),\]
Then, since lexicographic ordering $<_L$ is regular and well-founded, the following relation is regular
\[\psi(x_0,\ldots,x_n) \iff \phi(x_0,\ldots,x_n) \wedge \nexists y:\left[ y <_L x_n \wedge \phi(x_0,\ldots,y)\right],\]
And also a function. Restrictions of this sort are called \emph{Skolem functions}.

It is also worth noting that a classification of the regular functions also yields a classification of the regular languages, since for any regular language $R$, the characteristic function \[\mathcal{X}_R: w \mapsto \begin{cases} 1 & w \in R \\ 0 & w \notin R \end{cases}\]
Is regular.

The idea of achieving quantifier elimination on the regular word logic via function composition was inspired in part by a theorem in \cite{sql}, which decomposes regular predicates in terms of shifting operations, character-wise operations, and a univeral regular predicate.

Finally, in this section, I connect Moore Machines and Finite Automata:

\begin{prop}
Given a Moore Machine $M = (\Sigma, Q, q_0, \Gamma,\delta,\epsilon)$, the functions $M, M^\trunc,$ and $M^\rest$ are regular.
\end{prop}

\begin{proof}
Construct a deterministic finite automaton that keeps track of two pieces of information: the state the Moore Machine is expected to be in at any particular point, and a boolean value to keep track of whether the proposed output has so far been correct. Additionally, in the case of $M$, there will be a final character of the form $\colv{\#\\o}$ and a small amount of information must be kept track of to handle this correctly.

For example, for $M^\trunc$, let: \[A = (\Sigma \times \Gamma, Q \times \{0,1\}, (q_0,0), \delta', Q \times \{0\}),\]
Where:
\[\delta'((q,i),(a,o)) = \left(\delta(q,a), \begin{cases} 0 & i = 0 \wedge \epsilon(q) = o\\ 1 & \text{otherwise}\end{cases}\right).\]
\end{proof}

\begin{prop}
Every strictly causal, length-preserving, regular function is given by $M^\trunc$ for some Moore Machine $M$. Every causal, length-preserving function is given by $M^\rest$ for some Moore Machine $M$.
\end{prop}

\begin{proof}
I prove this for a binary, strictly causal, length-preserving, regular function $f$. The proof is nearly identical in the general case.
Let $f: \Sigma \to \Gamma$ be given. By Proposition \ref{wordlogic}, there is a formula $\phi(w_0,w_1)$ in the language $\mathcal{W}_{\Sigma \cup \Gamma}$ such that \[\phi(w_0,w_1) \iff f(w_0) = w_1.\]

Since $f$ is a strictly causal function, the first $n-1$ characters of the input determine the $n$th character of the output. By simple tricks in $\mathcal{W}_{\Sigma \cup \Gamma}$, one can construct a formula $\phi_u$ for each $u \in \Gamma$ such that $\phi_u$ is true of exactly those sequences of characters that produce an output of $u$ in the next place. These $\phi_u$ describe a collection of regular sets which partition all of $\Sigma^*$. Let $A_u$ be a finite automaton that recognizes the corresponding collection.

Now to construct the Moore Machine $M$. It should run each of the $A_u$ in parallel to determine its state. Since the $A_u$ recognize disjoint collections, exactly one of the $A_u$ will be in an accept state at any time. The $\epsilon$ function for the Moore Machine should take in the tuple of states for the $A_u$ and output the one which is in an accept state. It suffices now to check that $M^\trunc$ is identically $f$.

The proof for $M^\rest$ is similar.
\end{proof}

Of course, there are plenty of other regular functions. In this paper I provide a small set of functions whose closure under multi-ary composition generates all of them. In the next section, I will interpret the proof of Proposition \ref{detlemma} to reduce the problem to studying functions given by the actions of Moore Machines and Reverse Moore Machines. Moore Machines allow one to construct functions which transmit information only to the right (towards the end of the inputs), but they need to be combined with a method to transmit information to the left. Consider, for instance, the $\rest$ function, which removes the first character of a string, shifting to the left. Alternately, a function which outputs $00$ or $01$ depending on whether there are an even or odd number of $0$s in the input. These functions are regular, but not causal. 

\section{Determinization and Harvesting}

In this section, I define Reverse Moore Machines, and provide a technique for decomposing a length-preserving regular function as a multivariable composition of a Moore Machine function, a Reverse Moore Machine function, and character-wise maps to connect them. This will require several stages. First, I will briefly discuss some notation for discussing character-wise functions. Second, I will introduce the notion of a Reverse Moore Machine similar to the Moore Machines introduced in section 1. Third, I will present the decomposition. This decomposition is based on the classical powerset determinization construction, viewed from a novel perspective. In the next section, I will discuss handling general regular functions.

First, some notation for functions which operate character-wise:

\begin{defi}
Given two finite alphabets $\Sigma, \Gamma$, and a function $f: \Sigma \to \Gamma$, call the function $\cw_f: \Sigma^* \to \Gamma^*$ which applies $f$ to each character of the input, the \emph{character-wise $f$ map}.

If $f$ is an $n$-ary function for $n > 1$, one can also make sense of $\cw_f$. Let \[f: \Sigma_0 \times \cdots \times \Sigma_{n-1} \to \Gamma.\] Then one defines the partial function \[\cw_f: \Sigma_0^* \times \cdots \times \Sigma_{n-1}^* \to \Gamma^*,\] which takes in inputs of all the same lengths and produces an output of the same length, where:
\[\cw_f(w_0,\ldots,w_{n-1}) = \cw_f(\tuplefy(w_0,\ldots,w_{n-1})).\]
Noting that $\tuplefy$ does not produce characters with $\#$ in them for equal-length inputs.
\end{defi}

Of course, $\cw_f$ is causal and length preserving (and reverse-causal, when I define the notion). Every Moore Machine function can be written as the action of the corresponding transparent Moore Machine composed with a bitwise application of its $\epsilon$ function.

I now define some notation for dealing with Reverse Moore Machines, analogous to the notation established previously.

\begin{defi}
A \emph{Reverse Moore Machine} is a tuple: $(\Sigma, Q, q_f, \Gamma, \delta, \epsilon)$.
\begin{itemize}
\item $\Sigma$ is a finite set of input characters (alphabet).
\item $Q$ is a finite set of states.
\item $q_f \in Q$ is the final state.
\item $\Gamma$ is a finite set of output characters.
\item $\delta: Q \times \Sigma \to Q$ is the reverse transition function.
\item $\epsilon: \Sigma \to \Gamma$ is the output function.
\end{itemize}

Given an input word $w \in \Sigma^*$, construct a run $r$ and output $o$ such that:
\begin{itemize}
\item $r[|w|] = q_f$.
\item $r[i] = \delta(r[i+1],w[i])$ for $0 \leq i < |w|$.
\item $o[i] = \epsilon(r[i])$, for $0 \leq i \leq |w|$.
\end{itemize}
\end{defi}

\begin{exa}
Here I show below how inputs, states, and outputs, respectively, line up according to this notation.
\begin{center}\includegraphics{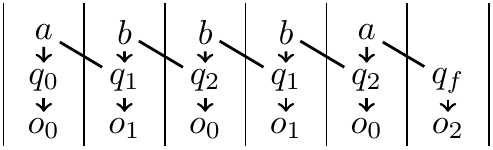} \end{center}
\end{exa}

To prevent type mismatches, I will denote Reverse Moore Machines with letters $R,P$ as opposed to letters $M,N$ for Moore Machines.

\begin{defi}
A Reverse Moore Machine is said to be \emph{transparent} if $\Gamma = Q$ and $\epsilon$ is the identity.
In this case, the Reverse Moore Machine is presented as a tuple: $(\Sigma, Q, q_f, \delta)$.
\end{defi}

As before, one can interpret Reverse Moore Machines as functions from input words to output words:
\begin{defi}
Given a Reverse Moore Machine $R = (\Sigma, Q, q_f, \Gamma, \delta, \epsilon)$, and word $w \in \Sigma^*$, denote by:
\begin{itemize}
\item $R(w)$ the output of $R$ on input $w$.
\item $R^\trunc(w)$ the output of $R$ on input $w$ with the last state removed.
\item $R^\rest(w)$ the output of $R$ on input $w$ with the first state (the start state) removed.
\end{itemize} 
\end{defi}

Analogous to the notions of causal and strictly causal, one has notions of reverse causal and strictly reverse causal:

\begin{defi}
Given two alphabets, $\Sigma, \Gamma$, a function on words $f: \Sigma^* \to \Gamma^*$ is said to be:
\begin{itemize}
\item \emph{reverse causal} if $f(w)[i]$ depends only on $w[i], w[i+1], \ldots$.
\item \emph{strictly reverse causal} if $f(w)[i]$ depends only on $w[i+1], w[i+2], \ldots$.
\end{itemize}
\end{defi}

Of course, a function is character-wise iff it is causal and reverse causal.

One also has a notion of reverse deterministic automaton:
\begin{defi}
A finite automaton $A = (\Sigma, Q, I, \delta, F)$ is \emph{reverse deterministic} if:
\begin{itemize}
\item $F$ is a singleton.
\item For each $q, a$, there is a unique $q'$ such that $(q',a,q) \in \delta$.
\end{itemize}
\end{defi}

A reverse deterministic automaton can also be viewed as a transparent Reverse Moore Machine.

I now have the notation to state the main result:
\begin{thm}\label{detharvest}
Any length-preserving regular function can be written as a multivariable composition of the form:
\[f(w) = R^\trunc(\cw_{\pair}(w,M^\trunc(w))).\]
\end{thm}

It's worth noting that $\cw_{\pair}$ has the same action here as $\tuplefy$. I've chosen to use $\cw_{\pair}$ here to indicate that I'm not using the length-padding features of $\tuplefy$. 

A few lemmas must be proved first:

\begin{lem}\label{automatize}
Given a length-preserving regular function $f: \Sigma^* \to \Gamma$, there is an automaton $B$ with state set $Q$ and map $\epsilon: Q \to \Gamma$ such that $f(w)$ is $\trunc \circ \cw_\epsilon$ applied to any run of $B$ on input $w$.
\end{lem}

\begin{proof}
Given the length-preserving function $f: \Sigma^* \to \Gamma^*$, let: \[C = ((\Sigma \cup \{\#\}) \times (\Gamma \cup \{\#\}), Q, I, \delta, F)\] Witness the regularity of $f$. Since $C$ automatically rejects any inputs with a $\#$ in them, one can restrict it to:
\[A = (\Sigma \times \Gamma, Q, I, \delta, F)\]
Recognizing the same set of inputs.

Construct an automaton:
\[B = (\Sigma, Q \times \Gamma, I \times \Gamma, \delta', F \times \Gamma),\]
Where:
\[((q,o),a,(q',o')) \in \delta' \iff (q,(a,o),q') \in \delta\]
$B$ is nondeterministic. I assert that runs of $B$ on input $w$ will necessarily have $\Gamma$ component $S_o(f(w))$, for some $o \in \Gamma$. Firstly, it should be clear that if $r$ is an accepting run of $A$ on input $\tuplefy(w,f(w))$, then $\tuplefy(r,S_o(f(w)))$ is an accepting run of $B$ on input $w$ for any $o \in \Gamma$.

Now, suppose $\tuplefy(r',g)$ is an accepting run of $B$ on input $w$. Then it is easy to check that $r'$ is an accepting run of $A$ on input $\tuplefy(w,\trunc(g))$. Since $A$ witnessed $f$ being a regular function, $\trunc(g)$ must be $f(w)$. This completes the proof for $\epsilon$ taking the $\Gamma$ component of the states of $B$.
\end{proof}

\begin{defi}
Given a nondeterministic automaton $A = (\Sigma, Q, I, \delta, F)$, define its \emph{determinization}:
\[\det(A) = (\Sigma, \mathcal{P}(Q), \{I\}, \delta', \{E \subset Q: E \cap F \neq \emptyset\}),\]
Where:
\[\delta'(K,a) = \{q' \in Q | \exists q \in K: \delta(q,a,q')\}.\]
\end{defi}

Typically throughout this paper I will be concerned with the determinization as a transparent Moore Machine, so the set of final states doesn't matter much.

\begin{defi}
Given a nondeterministic automaton $A = (\Sigma,Q, I \delta, F)$, define its \emph{harvester}:
\[\harv(A) = (\Sigma \times \mathcal{P}(Q), Q \cup \{q_F\}, Q \cup \{q_F\}, \delta', \{q_F\}),\]
Where:
\begin{itemize}
\item The $q$ such that $\delta'(q,(a,K),q')$ (for $q' \in Q$) is given by the least $q'' \in K$ such that $\delta(q'',a,q')$ or $q_F$ if none exist.
\item The $q$ such that $\delta'(q,(a,K),q_F)$ is given by finding the least $q' \in F$ such that there is a $q'' \in K$ such that $\delta(q'',a,q')$ and then having $q$ be the least $q'' \in K$ such that $\delta(q'',a,q')$. If no such $q' \in F$ exists, then $q$ is given by $q_F$.
\end{itemize}
\end{defi}

The powerset determinization of an automaton $A$ produces an automaton that keeps track of, at every position, the set of states of $A$ which are reachable through some sequence of transitions, having read the input up to that point. However, not every one of these reachable states necessarily shows up in some accepting run: it may be that being in one state now means later on having to be in another state which it cannot transition out of, or that being in a state now dooms the automaton to being in a reject state once it has finished reading the input. In order to use the determinization to find an accepting run of the original automaton, provided there is an accepting run, one needs to start at the end and work backwards, all the while staying within states that one knows can be traced through a sequence of transitions back to a start state at the beginning.

Specifically, if the automaton is in a state $q$ which is reachable through some sequence of transitions after having read $S_a(w)$, then there must be at least one state $q'$ which is reachable through some sequence of transitions after having read $w$ such that reading $a$ takes it from state $q'$ to state $q$.

It is necessary to introduce an additional dummy state $q_F$ to start out in to make sure the harvester automaton is reverse-deterministic. Although this application of the harvester automaton will see only pairs of the form $(a,K)$ for $a$ some symbol being read by the original automaton and $K$ the set of states reachable immediately prior to reading that specific $a$, the automaton should be prepared to read in arbitrary input pairs. An invalid input will cause the reverse deterministic automaton to go into the dummy $q_F$ state. Additionally, a cheap fix is necessary to account for the fact that one doesn't know what state the run of the original automaton ends on and one needs a single ``final'' state for the reverse deterministic automaton to ``start out'' in. The dummy $q_F$ state represents all final states which are reachable. Truncating the run will remove this dummy state.

\begin{lem}\label{decomp}
Suppose one has a nondeterministic automaton $A$, and valid input $w$. Let $p$ be the run of $\det(A)$ on input $w$. I claim that $\harv(A)$ on input $\tuplefy(w,\trunc(p))$ will have run $r$, an accepting run of $A$ on input $w$.
\end{lem}

\begin{proof}
It suffices to show that the only occurrence of $q_F$ in $r$ is as the final character. By construction, $\harv(A)$ will satisfy the transition relations. As mentioned before the construction of $\harv(A)$ also prevents backwards transitioning into the $q_F$ state for this particular input, since a reachable state can always be traced back to a reachable state.
\end{proof}

By lemma \ref{automatize}, given a regular, length-preserving function $f$, there is a nondeterministic automaton $B$ which takes in a word $w$ and has a run $r$ which projects to $f(w)$. By Lemma \ref{decomp}: 
\[r = \harv(B)^{\trunc}(\cw_{\pair}(w,\det(B)^\trunc))\]
By modifying the outer Reverse Moore Machine, one can throw in the appropriate projection to its output map to produce $f(w)$. This completes the proof of Theorem \ref{detharvest}.

This is a remarkable result. Every length-preserving regular function can be computed in a two-step process: one pass forwards through the input leaving behind some information, then a pass backwards through the input with this additional information to directly produce the output. Two passes suffice; having more passes doesn't increase the expressive power.

\section{Length Modification}

The only functions this paper has dealt with so far were length-preserving. In this section, I show that most of the interesting behavior of regular functions was already captured in the length-preserving case.

\begin{prop}
Suppose $f: \Sigma \to \Gamma$ is a regular function. Then there is a fixed constant $c$ associated to $f$ such that $f(w)$ is no longer than $c + |w|$ for every $w$.
\end{prop}

The proof is based on the pumping lemma. 

\begin{proof}
Let $A$ be a deterministic automaton with $c$ states accepting exactly words of the form $\tuplefy(w,f(w))$. Choose a specific $w$ and suppose $f(w)$ is longer than $c + |w|$. Imagine what happens as $A$ reads in $\tuplefy(w,f(w))$, specifically, after $w$ is finished, and $A$ is reading in characters of the form $\colv{\#\\o}$ for some $o \in \Gamma$. Because there are more positions like this than there are states of $A$, by the pigeonhole principle some two positions will have the same state, say at positions $i$ and $j$. Note however that if one removes all positions in $\tuplefy(w,f(w))$ between $i$ and $j$ (including $i$, excluding $j$) one still has an accepting run, of the form $\tuplefy(w,g)$ for some $g$ strictly shorter than $f(w)$. This contradicts the assumption that $A$ accepted exactly words of the form $\tuplefy(w,f(w))$.
\end{proof}

From this, one can also show that for any $n$-ary regular function $f$, there is a fixed constant $c$ associated to $f$ such that $f(w_0,\ldots,w_{n-1})$ is no longer than $c$ plus the length of the maximum input.

Note that the automaton $A$ in the proof never encounters the input character $\colv{\#\\\#}$. As such, one may assume that this character acts as the identity on the states of $A$. The resulting automaton recognizes all pairs of the form $\tuplefy(S_\#^b(w),f(w))$ for arbitrary $b$ ($S_\#^c$ simply represents a $c$-fold composition of the $S_\#$ function). By adding in a counter, one can recognize exactly the pairs of the form $\tuplefy(S_\#^c(w),f(w))$.

As such, the function which takes $S_\#^c(w)$ to $S_\#^d(f(w))$ for $d = |f(w)| - |w|$ is a length-preserving regular function $g$, and thus can be written as a composition of character-wise maps, truncated Moore Machines, and truncated Reverse Moore Machines as in Theorem 1.
Now $f$ can be written as:
\[f(w) = \unpad(g(S_\#^c(w))),\]
Where $\unpad$ removes final $\#$ characters. Note that one does not know how many final $\#$ characters there will be. For functions which reduce length, the number will be more than $c$ and could be as much as $c+|w|$. One might be tempted to try to replace $\unpad$ with some function like $S_\#^{-1}$ (or, even less suited to the task, $\trunc$), which either removes a single final $\#$ or leaves the word alone if it cannot. However, since there are regular functions which take words of arbitrary length and reduce them to length 1, the generating set needs a generator that can produce unbounded shortening as well. 

Note that if $f$ is $n$-ary for $n > 1$, it follows that:
\[f(w_0,\ldots,w_{n-1}) = \unpad(g(\tuplefy(S_\#^c(w_0),\ldots,S_\#^c(w_{n-1})))),\]
For some regular, length-preserving function $g$.

As such:

\begin{thm}
Any regular function can be written as a multivariable composition of:
\begin{itemize}
\item Truncated Moore Machines,
\item Truncated Reverse Moore Machines,
\item Character-wise maps,
\item $\tuplefy$ (allowing one to generate multiary character-wise maps),
\item $S_a$ for various $a$,
\item $\unpad$.
\end{itemize}
\end{thm}

It's worth noting here that this generating set is infinite. Specifically, there are an infinite number of Moore Machines and Reverse Moore Machines. There are also an infinite number of Character-wise maps, but this isn't essential -- one could use encoding methods to work purely with a single two-character alphabet (plus, optionally, the dummy character $\#$). 

The infinitude of the generating set, however, is essential. The easiest way to see this is to talk about period introduction. Provided the input to a regular function has a sufficiently long periodic portion in the middle, the output of the regular function will also have a long periodic portion in the middle (it suffices to verify this of the generators above). What's more, the period of the periodic portion of the output can only have prime factors which show up either in the periods of the periodic portions of the inputs or which are smaller than the number of states of the associated automaton to the regular function. However, one can easily build regular functions which introduce any prime factor into the periodicity of their inputs, so there must not be any bound on the sizes of associated automata to regular functions in the generating set. In summary:

\begin{prop}
Any set of regular functions which generates all regular functions under multivariable composition must be infinite.
\end{prop}

I conclude this paper with a discussion of known results regarding the Krohn-Rhodes Theorem. First, I will provide a proof of the Krohn-Rhodes Theorem adapted from Ginzburg \cite{ginzburg}. Then I will use the Krohn-Rhodes Theorem to decompose the Moore Machine and Reverse Moore Machine generators into smaller, simpler generators. The idea of interpreting a the cascade given by the Krohn-Rhodes theorem as a composition of functions can be found in \cite{eilenberg}. This paper will spell out this composition precisely and in modern notation, as has been done for previous compositions. A few sections will be dedicated to cleaning up the resulting set of generators, followed by proposed future research.

\section{The Krohn-Rhodes Theorem}

In this section, I present the Krohn-Rhodes Theorem as adapted to the context of multivariate composition of regular functions. The original proof of the Krohn-Rhodes Theorem, in \cite{krohn}, was presented in terms of wreath products of semigroups. More modern presentations of the Krohn-Rhodes Theorem typically present it in terms of the cascade product of finite state transducers. 

The cascade product of two transducers $M_1, M_2$ is a single system consisting of both machines. First, machine $M_2$ reads in both the input to the system and the current state of $M_1$ to update its state. When it has finished, machine $M_1$ updates its state based only on the input to the system. Finally, an output is produced based on the states of $M_1$ and $M_2$. This reflects the reality of systems where updating the states of machines takes a small but appreciable amount of time. In a well designed system, $M_2$ should not have to wait for $M_1$ to finish its update before it can update its state. As such, $M_2$ uses the state of $M_1$ prior to reading the input to update.

The Krohn-Rhodes Theorem separates out two extremes of behavior for finite automata. In general, reading in an input character induces a function on the states of the automaton. This function may map two states to the same state or to separate states. At one extreme, it may act as a \emph{permutation} in which case it maps all states to separate states. In this case, it is possible to undo this action. One can recover the state before reading a character which acts as a permutation, provided one knows which character the automaton read. At the other extreme, an input character may act as a \emph{reset} in which case it maps all states to the same state. In this case all information about the previous state is lost.

\begin{defi}
A Moore Machine or deterministic automaton is said to be:
\begin{itemize}
\item A \emph{permutation automaton} if each of its inputs acts as a permutation on its states,
\item A \emph{reset automaton} if each of its inputs acts as a reset or the identity on its states,
\item A \emph{permutation-reset automaton} if each of its inputs acts as a permutation or a reset on its states.
\end{itemize}
\end{defi}

I now state the Krohn-Rhodes theorem, in a bit of an unusual fashion:

\begin{thm}[Krohn-Rhodes]
Given a transparent Moore Machine $M$, one can write its truncated action $M^\trunc$ as a multivariable composition of truncated actions of permutation-reset Moore Machines $M_0, \ldots, M_{n-1}$ for $n$ the number of states of $M$, and a final character-wise map $f$. What's more, this composition takes on a fairly simple form. Let:
\begin{align*}
w_0 &= M_0^\trunc(w),\\
w_1 &= M_1^\trunc(\tuplefy(w,w_0)),\\
w_2 &= M_2^\trunc(\tuplefy(w,w_0,w_1)),\\
&\vdots\\
w_{n-1} &= M_{n-1}^\trunc(\tuplefy(w,w_0, \ldots, w_{n-2})),
\end{align*}
Then:
\[M^\trunc(w) = \cw_f(w_0,\ldots,w_{n-1}).\]
\end{thm}

Every map in this composition is length-preserving. Of course one can write this as simply one large multivariable composition of a character-wise map and truncated actions of permutation-reset transparent Moore Machines, but this is unwieldy to write down. The above also presents an efficient way of computing the composition, although readers concerned with efficiency are encouraged to look into the Holonomy decomposition \cite{nagy05}.

The proof is inductive: I show that for every transparent Moore Machine $M$, there's another transparent Moore Machine $\overline{M}$ that keeps track of a state $M$ is not in in a permutation-reset way. This reduces the amount of information that needs to be kept track of by one state, and one can keep doing this until one has kept track of all the information to know what state $M$ is in.

The proof in \cite{ginzburg} allows for the possibility that one can keep track of several states the automaton $M$ is not in in a permutation-reset way at the same time, as opposed to one at a time in the proof below.  This is more efficient, but adds needless complexity to the proof.

The key piece of the construction is the Permutation-Reset Lemma, below.

\begin{lem}[Permutation-Reset Lemma]
Given two finite ordered sets of the same size $I,J$, and map between them $f$, there is a map $g: I \to J$ such that:
\begin{itemize}
\item $g$ either acts as:
\begin{itemize}
\item A bijection from $I$ to $J$ (a permutation on the position indices),
\item Or has singleton image (a reset on position indices),
\end{itemize}
\item And for $x, y \in I$, with $x \neq y$,$f(x) \neq g(y)$.
\item For any $x \in I$, $f$ maps elements of $I \setminus \{x\}$ to $J \setminus \{g(x)\}$.
\end{itemize}
\end{lem}

\begin{proof}
Suppose $f$ does act as a bijection. Then $g = f$ is a permutation and satisfies the inequality condition.

Suppose $f$ does not act as a bijection. Then $g$ which maps everything to the smallest element of $J$ which is not in the image of $f$ has singleton image and satisfies the inequality condition.

The third condition is just a rephrasing of the second, but will come in handy later on.
\end{proof}

A specific application of the Permutation-Reset Lemma is that one can have a transparent Moore Machine that keeps track of a state the original transparent Moore Machine is not in:

\begin{lem}
Given a transparent Moore Machine $M = (\Sigma, Q, q_0, \delta)$, there is a permutation-reset transparent Moore Machine $\overline{M}$ with the same state set such that on input $w$, the state of $M$ at any one time is not the state of $\overline{M}$.
\end{lem}

\begin{proof}
Assign a natural ordering to $Q$. Define \[\overline{M} = (\Sigma, Q, \overline{q}_0, \overline{\delta}),\]
Where $\overline{q}_0$ is the smallest element in $Q$ which is not $q_0$, and $\overline{\delta}(q,a)$ is given by:
\begin{itemize}
\item $\delta(q,a)$ if $a$ acts as a permutation.
\item Otherwise, the smallest $q' \in Q$ which is not in the image of any state under the action of $a$.
\end{itemize}
If the action of $a$ was a permutation originally, it is still a permutation in the new automaton. This permutation not only maps the state the automaton is in before reading $a$ to the state afterwards, but also from a state the automaton is not in before reading $a$ to a state the automaton is not in afterwards. In the second case, note that the choice of $q'$ does not depend on $q$, so this action is a reset. Obviously, it maps a state the original automaton is not in before reading $a$ to a state the automaton is not in after reading $a$.
\end{proof}

\begin{lem}
Given transparent Moore Machines $M, \overline{M}$ as above with state sets $Q$, there is a third transparent Moore Machine:
\[\widehat{M} = (\Sigma \times Q, \{0,\ldots,|Q|-2\}, \hat{\imath}_0, \widehat{\delta}),\]
Such that for any input $w$, if $M$ on reading $w$ winds up in state $q$, and $\overline{M}$ on reading $w$ winds up in state $\overline{q}$ then $\widehat{M}$ on reading $\tuplefy(w,{\overline{M}}^{\trunc}(w))$ will wind up in state $\hat{\imath}$, where $q$ is the element in position\footnote{To maintain notational consistency within this paper, where ordered collections are indexed starting with 0, I refer to the first element of a set as being in position 0, and generally the $i+1$st element of a set as being in position $i$. To avoid confusion, I avoid using the notation ``$i$th element'', instead using notation ``the element in position $i$''.
} $\hat{\imath}$ of $Q \setminus \{\overline{q}\}$.
\end{lem}

\begin{proof}
Let $\widehat{M} =  (\Sigma \times Q, \{0,\ldots,|Q|-2\}, \hat{\imath}_0, \widehat{\delta})$ where $\hat{\imath}_0$ is the index of $q_0$ in $Q \setminus \{\overline{q}_0\}$ and $\widehat{\delta}(i,(a,\overline{q}))$ is computed by:
\begin{itemize}
\item Computing $q' = \overline{\delta}(\overline{q},a)$. This is the state $\overline{M}$ says that $M$ is not in after reading $a$.
\item Computing $q$, the $i$th element of $Q \setminus \{\overline{q}\}$. This is the state of $M$ prior to reading $a$.
\item Return $j$, the index of $\delta(q,a)$ in $Q \setminus q'$.
\end{itemize}

The above construction is designed specifically to satisfy the conclusion.
\end{proof}

As such: \[M^{\trunc}(w) = \cw_p({\overline{M}}^{\trunc}(w),{\widehat{M}}^{\trunc}(\tuplefy(w,{\overline{M}}^{\trunc}(w))))\]
Where $p(\overline{q},i)$ is the $i$th element of $Q \setminus \{\overline{q}\}$.

Finally, I prove the Krohn-Rhodes Theorem:
\begin{proof}
By induction on the number of states of $M$.

\textbf{Base Case:} If $M$ has one state, then $f$ in the composition is 0-ary, and one can have it just output that constant state.

\textbf{Inductive Case:} Given $M$, one can write: 
\[M^{\trunc}(w) = \cw_p({\overline{M}}^{\trunc}(w),{\widehat{M}}^{\trunc}(\tuplefy(w,{\overline{M}}^{\trunc}(w)))),\]
With $\overline{M}$ permutation-reset. By the inductive hypothesis, one can write ${\widehat{M}}^{\trunc}(w)$, which has one fewer state than $M$, as:
\[{\widehat{M}}^\trunc(w) = \cw_f(w_1,\ldots,w_{n-1}),\]
Where:
\begin{align*}
w_1 &= M_1^\trunc(w)\\
w_2 &= M_2^\trunc(\tuplefy(w,w_1)))\\
w_3 &= M_3^\trunc(\tuplefy(w,w_1,w_2)\\
&\vdots\\
w_{n-1} &= M_{n-1}^\trunc(\tuplefy(w,w_1, \ldots, w_{n-2}))
\end{align*}

As such, ${\widehat{M}}^{\trunc}(\tuplefy(w,{\widehat{M}}^{\trunc}(w)))$ is given by just plugging in:

\[{\widehat{M}}^{\trunc}(\tuplefy(w,{\overline{M}}^{\trunc}(w))) = \cw_f(w_1,\ldots,w_{n-1}),\]
Where:
\begin{align*}
w_1 &= M_1^\trunc(\tuplefy(w,{\overline{M}}^{\trunc}(w))),\\
w_2 &= M_2^\trunc(\tuplefy(\tuplefy(w,{\overline{M}}^{\trunc}(w)),w_1)),\\
w_3 &= M_3^\trunc(\tuplefy(\tuplefy(w,{\overline{M}}^{\trunc}(w)),w_1,w_2)),\\
&\vdots\\
w_{n-1} &= M_{n-1}^\trunc(\tuplefy(\tuplefy(w,{\overline{M}}^{\trunc}(w)), w_1, \ldots, w_{n-2})).
\end{align*}

Alternately, fiddling with some parentheses in the definitions in the automata:
\[{\widehat{M}}^{\trunc}(\tuplefy(w,{\overline{M}}^{\trunc}(w))) = \cw_f(w_1,\ldots,w_{n-1},)\]
Where:
\begin{align*}
w_0 &= {\overline{M}}^{\trunc}(w),\\
w_1 &= M_1^\trunc(\tuplefy(w,w_0)),\\
w_2 &= M_2^\trunc(\tuplefy(w,w_0,w_1)),\\
w_3 &= M_3^\trunc(\tuplefy(w,w_0,w_1,w_2)),\\
&\vdots\\
w_{n-1} &= M_{n-1}^\trunc(\tuplefy(w,w_0, \ldots, w_{n-2})).
\end{align*}
In which case:
\[M^{\trunc}(w) = \cw_p(\underbrace{{\overline{M}}^{\trunc}(w)}_{w_0},\cw_f(w_1,\ldots,w_{n-1})).\]
One can combine $p$ and $f$ to get a single character-wise function on $w_0,\ldots, w_{n-1}$, thus completing the induction.
\end{proof}

The proof is still straightforward if one unwinds the induction. In the construction, $w_0$ is keeping track of a state $M$ is not in, but it may as well be keeping track of an index for a state $M$ is not in. $w_1$ is keeping track of an index of a state $M$ is not in once one has removed the state in position $w_0$ from $Q$. $w_2$ is keeping track of an index of a state $M$ is not in once one has removed the states that $w_1$ and $w_2$ are keeping track of from $Q$. And so forth. One can formalize this indexed removal process as follows:

\begin{defi}
Given a positive integer $n$, an \emph{ordinal removal sequence} for $n$ is a (possibly empty) sequence of positive integers $(k_0,\ldots,k_i)$ satisfying:
\[ i < n,\]
\[0 \leq k_j < n-j.\]
\end{defi}

One can interpret an ordinal removal sequence for $n$ as a series of commands operating on an ordered set of size $n$ of the form ``remove the $i+1$st smallest element remaining.'' Note that as elements are removed, there are fewer elements remaining, hence the decreasing upper limit on $k_j$ in the second constraint. Consistent with the rest of this paper, I begin indexing with 0, so a 0 means remove the smallest element.

\begin{defi}
Given an ordered set $L$, and an ordinal removal sequence for $|L|, \mathbf{k} = (k_0,\ldots,k_{i-1})$, define $\remove(L, \mathbf{k})$ recursively:
\begin{itemize}
\item $\remove(L,()) = L$,
\item $\remove(L,(k_0,\ldots,k_{i-1}))$ is given removing the element in position $k_i$ (with order inherited from $L$, and 0 means remove the smallest element) of $\remove(L,(k_0,\ldots,k_{i-2}))$.
\end{itemize}
Let $\ors_{k,n}$ denote the set of ordinal removal sequences for $n$ of length $k$.
\end{defi}

Recall that ${}^\frown$ is used for concatenation, so: \[(k_0,\ldots,k_{i-1})^{\frown}j = (k_1,\ldots,k_{i-1},j).\]

\begin{exa}
Consider the ordinal removal sequence for 5: $(0,1,2,1)$ acting on the ordered set $(A,B,C,D,E)$:
\begin{center}
\begin{tabular}{|c|l|}
\hline
Start: & $(A,B,C,D,E)$\\
Remove at position 0: & $(B,C,D,E)$\\
Remove at position 1: & $(B, D, E)$\\
Remove at position 2: & $(B, D)$\\
Remove at position 1: & $(B)$\\
\hline
\end{tabular}
\end{center}
\end{exa}

The following lemmat formally gives the construction of the $M_j$ in the Krohn-Rhodes decomposition.

\begin{lem}
Given a transparent Moore Machine $M = (\Sigma, Q, q_0, \delta)$, and $0 \leq j < |Q|-1$, there is a permutation-reset transparent Moore Machine \[M_j = (\Sigma \times \ors_{j,|Q|}, \{0,\ldots,|Q|-j\}, k_0, \delta'),\]
Such that if $\kappa$ is a word of ordinal removal sequences for $|Q|$ of length $j$, that is, $\kappa \in (\ors_{j,|Q|})^*$, satisfying:
\begin{itemize}
\item At any point $i$, $\kappa[i]$ does not remove $M(w)[i]$ from $Q$. That is, $M(w)[i] \in \remove(Q,\kappa[i])$.
\item $\kappa[i+1]$ is determined by $\kappa[i]$ and $w[i]$, specifically such that:
\item The map $\delta_a$ (from $M$) maps states in $\remove(Q,\kappa[i])$ to states in $\remove(Q,\kappa[i+1])$.
\end{itemize}
Then $o = M_j(\tuplefy(w,\kappa))$ satisfies:
\begin{itemize}
\item At any point $i$, $\kappa[i] {}^{\frown} o[i]$ does not remove $M(w)[i]$ from $Q$. That is, $M(w)[i]$ is not in position $o[i]$ of $\remove(Q,\kappa[i])$, and in particular $M(w)[i] \in \remove(Q,\kappa[i]{}^{\frown}o[i])$.
\end{itemize}
\end{lem}

\begin{proof}
The idea here is that each $M_j$ should keep track of a single entry in an ordinal removal sequence that will remove all elements of $Q$ except the state of the original automaton $M$ at any one particular time. These will be the $M_j$ in the multivariate composition, so they will be reading in both the original input (a single character $a$ from $w$), and a single index from each of $M_0,\ldots,M_{j-1}$, together forming an ordinal removal sequence ($\kappa[i]$) of length $j$. Each $M_j$ then keeps track of an index, which, when added on to the end of the ordinal removal sequence does not remove the one state that must not be removed, the state of $M$ at that point.

For reference, a picture of the situation is drawn below:
\begin{center} \includegraphics{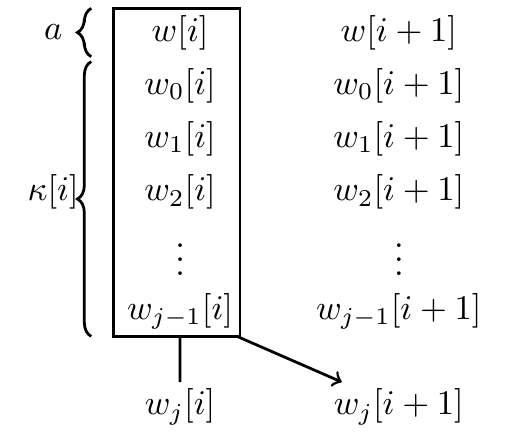} \end{center}

As one can see in the diagram, $M_j$ will be reading in $w[i]$, the character that takes the automaton $M$ from $M(w)[i]$ to $M(w)[i+1]$, and $\kappa[i]$, the ordinal removal sequence within which its current state $w_j[i]$ is interpreted. Specifically, $w_j[i]$ will be a position in $\remove(Q,\kappa[i])$ where there isn't the current state of $M$, $M(w)[i]$.
This will transition $M_j$ into the state $w_j[i+1]$, which must be a position in $\remove(Q,\kappa[i+1])$ where there isn't the next state of $M$, $M(w)[i]$.

Note that $M_j$ does not get direct access to $\kappa[i+1]$, but of course it needs access to $\kappa[i+1]$ in order to determine the index for $M(w)[i+1]$ in $\remove(Q,\kappa[i+1])$ so that it can avoid it. Fortunately, if the previous automata, $M_0,\ldots,M_{j-1}$ work in canonical fashions, knowing $a$ and $\kappa[i]$ is enough to determine $\kappa[i+1]$.

To start with, one needs to pick the starting state ($w_j[0]$) for $M_j$, $k_0$, such that $k_0$ is not an index for the start state of $M$ in $\remove(Q,\kappa[0])$. Let it be the smallest such index.

By hypothesis, the transition map induced by the character $a$ on the automaton $M$, $\delta_a$, maps states in $\remove(Q,\kappa[i])$ to states in $\remove(Q,\kappa[i+1])$. By the Permutation-Reset Lemma, one can define the transition map for $M_j$, $\delta'$, with $\delta'_{(w[i],\kappa[i])}$
a permutation-reset map for any particular $w[i]$ and $\kappa[i]$ that does the avoiding required of it.
\end{proof}
Finally, it's worth noting that $\kappa[i]^{\frown}w_j[i]$, $\kappa[i+1]^{\frown}w_j[i+1]$ satisfy the requirements on $\kappa[i]$ and $\kappa[i+1]$ in the hypothesis of this lemma. Specifically:
\begin{itemize}
\item The new $w_j[0]$ and $w_j[i+1]$ were chosen to avoid removing $M(w)[0]$ and $M(w)[i+1]$ from $Q$.
\item $w_j[i+1]$ is determined by $w_j[i], \kappa[i]$, and $w[i]$.
\item The map $\delta_a$ maps states in $\remove(Q,\kappa[i]^{\frown}w_j[i])$ to states in $\remove(Q,\kappa[i+1]^{\frown}w_j[i+1])$. This is immediate, looking at the third condition on the function the Permutation-Reset Lemma constructs.
\end{itemize}

As such the $M_j$ in this proof are the same as the $M_j$ in the multivariate composition for the action of $M$, and $\tuplefy(w_1,\ldots,w_{j-1})$ is a word of ordinal removal sequences $\kappa$ as above for each suitable $j$. As such the final character-wise map $f$ is simply the map mapping an ordinal removal sequence $\mathbf{k}$ of length $|Q|-1$ to the single element of $\remove(Q,\mathbf{k})$.

Hopefully, this particular proof will shed some light on the multivariate composition used to determine $M^{\trunc}(w)$: why it is shaped the way it is shaped, and what each piece of the composition is keeping track of. Utilizing the Krohn-Rhodes Theorem, one currently has the following set of generators for the regular functions:

\begin{thm}
Any regular function can be written as a multivariable composition of:
\begin{itemize}
\item Truncated Permutation-Reset Transparent Moore Machine maps,
\item Truncated Reverse Moore Machine maps,
\item Character-wise maps,
\item $\tuplefy$ (allowing one to generate multiary character-wise maps),
\item $S_a$ for various $a$,
\item $\unpad$.
\end{itemize}
\end{thm}

It may seem like this hasn't gained much, but the next section will show that there isn't actually that much to Permutation-Reset Transparent Moore Machines. The section after will handle the reverse Moore Machine case.

\section{A Further Breaking Down}

In this section I prove that truncated permutation-reset transparent Moore Machine maps can be written as the composition of a single truncated permutation transparent Moore Machine map and a single truncated reset transparent Moore Machine map. Then I show that permutation transparent Moore Machines and reset transparent Moore Machines are actually quite familiar objects. As before, these proofs are adapted from \cite{ginzburg}, which uses vastly different notation.

\begin{lem}
Given a permutation-reset transparent Moore Machine $M = (\Sigma, Q, q_0, \delta)$, there is a permutation transparent Moore Machine $\tilde{M}$, reset transparent Moore Machine $\vec{M}$, and function $F$ such that: 
\[M^{\trunc}(w) = \cw_F(\vec{M}^{\trunc}(\tuplefy(w,\tilde{M}^{\trunc}(w))),\tilde{M}^{\trunc}(w))\]
\end{lem}

\begin{proof}
Let $\tilde{M} = (\Sigma, S_Q, id, \tilde{\delta})$ and $\vec{M} = (\Sigma \times S_Q, Q, q_0, \vec{\delta})$, where $S_Q$ is the set of all permutations on $Q$, and:

If $\delta_a$ is a permutation of the states of $M$:
\[\tilde{\delta}(f,a) = \delta_a \circ f,\]
\[\vec{\delta}(q,(a,f)) = q.\]

If $\delta_a$ is a reset on the states of $M$ with image $\{q_a\}$:
\[\tilde{\delta}(f,a) = f,\]
\[\vec{\delta}(q,(a,f)) = f^{-1}(q_a).\]

As desired $\tilde{M}$ is a permutation automaton and $\vec{M}$ is a reset automaton (notice that the identity action is necessary in case $\delta_a$ is a permutation). 

Let $F: S_Q \times Q \to Q$ with $F(f,q) = f(q)$. I now claim that 
\[M^{\trunc}(w) = \cw_F(\vec{M}^{\trunc}(\tuplefy(w,\tilde{M}^{\trunc}(w))),\tilde{M}^{\trunc}(w)),\]
As desired. This is easy to verify in terms of their transition relations. The intuition behind this construction is that $\tilde{M}$ keeps track of the action of each of the permutations and $\vec{M}$ handles resets by storing them in terms of what state one would have to start in such that after being acted on by just the permutations one winds up in the current state of $M$.
\end{proof}

I now define a couple of transparent Moore Machines in order to refine the decomposition further.
\begin{defi}
For each $n$, define the \emph{Accumulator on $S_n$} transparent Moore Machine $AS_n$:
\[AS_n = (S_n,S_n,id,\delta),\]
Where:
\[\delta_g(h) = h \cdot g,\]
Where $S_n$ is the symmetric group on $n$ elements with composition operation $(h \cdot g)(i) = h(g(i))$.
\end{defi}

\begin{defi}
Define the \emph{bit-storage} automaton:
\[Bit = (\{-,0,1\}, \{0,1\}, 0, \delta),\]
Where $\delta_{-}$ acts as the identity, $\delta_0$ is a reset to state $0$, and $\delta_1$ is a reset to state $1$.
\end{defi}

Since every collection of permutations can be viewed as a subset of the symmetric group $S_n$ for some $n$:
\begin{prop}
Every truncated permutation transparent Moore Machine map $M^{\trunc}$ can be written as: 
\[M^{\trunc}(w) = \cw_f(AS_n^{\trunc}(\cw_g(w))),\]
For some functions $f$ and $g$.
\end{prop}

What's more:
\begin{prop}
Every truncated reset transparent Moore Machine map $M^{\trunc}$ can be written as:
\[\cw_f(Bit^{\trunc}(\cw_{g_0}(w)), \ldots, Bit^{\trunc}(\cw_{g_{n-1}}(w))),\]
For suitable $n$, $f$, and $g_0,\ldots,g_{n-1}$.
\end{prop}

\begin{proof}
Suppose $M = (\Sigma,Q,q_0,\delta)$. Choose an $n$ such that $2^n \geq |Q|$. For every state $q \in Q$, associate a unique bitstring $b(q) \in 2^n$, where $b_k(q)$ is the bit in position $k$ of $b(q)$, such that the start state $q_0 \in Q$ is given by the all 0s bitstring. Let $f = b^{-1}.$ Suppose $\delta_a$ acts as a reset to the state $q_a$. Then let $g_k(a) = b_k(q_a)$.
\end{proof}

As such:
\begin{prop}
Every truncated Moore Machine map $M^{\trunc}$ can be written as a multivariable composition of:
\begin{itemize}
\item $AS_n^{\trunc}$ for various $n$,
\item $Bit^{\trunc}$,
\item Character-wise maps.
\end{itemize}
\end{prop}

This yields a much smaller generating set for the regular functions:

\begin{thm}
Any regular function can be written as a multivariable composition of:
\begin{itemize}
\item $AS_n^{\trunc}$ for various $n$,
\item $Bit^{\trunc}$,
\item Truncated Reverse Moore Machine maps,
\item Character-wise maps,
\item $\tuplefy$ (allowing one to generate multiary character-wise maps),
\item $S_a$ for various $a$,
\item $\unpad$.
\end{itemize}
\end{thm}

\section{Reverse Moore Machines}

Just as in previous sections, I broke down the truncated Moore Machine maps into compositions involving the accumulator on $S_n$, the $Bit$ automaton, and character-wise maps (note that uses of $\tuplefy$ are length-preserving, and thus actually character-wise applications of a tuple-construction map),
in this section, I break down truncated reverse Moore Machine maps similarly. To save work, I will simply introduce a reversal map $\rev$ (which is not regular) to connect truncated reverse Moore Machine maps and truncated Moore Machine Maps.

\begin{defi}
Given a word $w \in \Sigma^*$ define $\rev(w)$ to be the reversal of $w$.

Given a Moore Machine $M = (\Sigma, Q, q_0, \Gamma, \delta, \epsilon)$, define its reversal:
\[\rev(M) = (\Sigma, Q, q_0, \Gamma, \delta, \epsilon).\]
And similarly define the reversal of a reverse Moore Machine.
\end{defi}

\begin{prop}
Given a Moore Machine $M$, and word $w$:
\[\rev(M)(w) = \rev(M(\rev(w))).\]
What's more:
\[\rev(M)^{\trunc}(w) = \rev(M^{\rest}(\rev(w))).\]
\end{prop}

Functions that are related in this way are said to be related by \emph{conjugation by $\rev$}. This relation is reflexive and symmetric. What's more since $\rev$ is its own inverse, if $f$ and $f'$ are related by conjugation and $g$ and $g'$ are related by conjugation, then $f \circ g$ and $f' \circ g'$ will be related by conjugation. Indeed this works for multiary functions as well:

\begin{defi}
Given an $n$-ary function $f$, say that \[(x_0,\ldots,x_{n-1}) \mapsto f(x_0,\ldots,x_{n-1})\] And
\[(x_0,\ldots,x_{n-1}) \mapsto \rev(f(\rev(x_0),\ldots,\rev(x_{n-1})))\]
Are related by \emph{conjugation by $\rev$}.
\end{defi}

\begin{prop}
Given a multi-ary composition of functions, if one replaces each function by its conjugation by $\rev$, the overall composition is related to the original composition by conjugation by $\rev$.
\end{prop}

\begin{proof}
It suffices to note that in the resulting composition, whenever the output of a function is fed into the input of another function, it is reversed twice, effectively doing nothing to it.
\end{proof}

Additionally, conjugation by $\rev$ does not alter character-wise functions.

It is necessary to prove that $M^{\rest}$ can be written in terms of $M^{\trunc}$:

\begin{prop}
Given a transparent Moore Machine $M = (\Sigma, Q, q_0, \delta)$, one can write $M^{\rest}$ as a composition of character-wise maps and $M^{\trunc}$.
\end{prop}

\begin{proof}
It is easy to verify that:
\[M^{\rest}(w) = \cw_\delta(M^{\trunc}(w),w).\]
\end{proof}

It follows from this that for any Moore Machine $M$, $M^{\rest}$ can be written as a composition of character-wise maps and $M^{\trunc}$.

\begin{prop}
Any truncated reverse Moore Machine map can be written as the multivariate composition of:
\begin{itemize}
\item $RAS_{n}^{\trunc}$, where $RAS_n$ is the reversal of $AS_n$, for various $n$,
\item $RBit^{\trunc}$, where $RBit$ is the reversal of $Bit$,
\item Character-wise maps.
\end{itemize}
\end{prop}

\begin{proof}
Every reverse Moore Machine is $\rev(M)$ for some $M$. As such, one can write $\rev(M)^{\trunc}$ as:
\[\rev \circ M^{\rest} \circ \rev.\]
One can write $M^{\rest}$ as a multiary composition of $AS_n^{\trunc}$ for various $n$, $Bit^{\trunc}$, and character-wise maps, so by the conjugation of compositions lemma, one can write $\rev(M)^{\trunc}$ as a multiary composition of the conjugations of those components.
\end{proof}

Also note:
\begin{prop}
Given a Moore Machine $M$, one can write $M(w)$ as the multivariable composition of a truncated Moore Machine map and $S_{\#}$.
\end{prop}

\begin{proof}
Augment $M$ to $M'$ by allowing it to interpret the input $\#$ (it may do so in any way it likes). Then:
\[M(w) = M^{\trunc}(S_{\#}(w)).\]
\end{proof}

As one final refinement of the generating set, I show that $RAS_n$ is unnecessary as a generator.

\section{Removing the Reverse Accumulator}

Note that $AS_n$ and $RAS_n$ are very similar automata. For $AS_n$, one interprets the input character $w[i]$ as a permutation relating $AS_n(w)[i]$ and $AS_n(w)[i+1]$. For $RAS_n$, one interprets the input character $w[i]$ as a permutation relating $RAS_n(w)[i+1]$ and $RAS_n(w)[i]$. Since every permutation has an inverse, shouldn't these two automata be the same up to a suitable character-wise map on the inputs? Alas, the distinction is more subtle: there is another constraint on the runs of $AS_n$ and $RAS_n$. For $AS_n$, the first character of its run is specified to be $id$. For $RAS_n$, the last character of its run is specified to be $id$.

Compare $AS_n(w)$ and $RAS_n(\cw_{inverse}(w))$ on some generic five character input $w = abcde$:

\[\begin{array}{c|cccccc}
w & a & b & c & d & e & \\
\hline
AS_n(w) & id & a & ab & abc & abcd & abcde \\
RAS_n(\cw_{inverse}(w)) & (abcde)^{-1} & (bcde)^{-1} & (cde)^{-1} & (de)^{-1} & e^{-1} &  id\\
\end{array}\]

In addition to applying a suitable transformation to the inputs of $AS_n$, one must also apply a suitable transformation to the outputs of $AS_n$ if one wants to produce the output of $RAS_n$. Specifically, if one multiplies every character in the output of $AS_n$ on the left by the inverse of the last character of the output, it will ensure that the new last character of the output is $id$, but still maintain the transition relationships. This requires passing the information of the last character of $AS_n$ to every other position.

First, I must introduce the \emph{mask} of the word $w$, a word $\mask(w)$ which is all 0s up to the length of $w$, followed by a 1. This is computed simply by taking $S_1(\cw_0(w))$ where 0 is the constant 0 map. Despite its simplicity, this word will be key to performing the computation.

Consider the reverse reset Transparent Moore Machine $R = (\{0,1\} \times S_n,S_n,id,\delta)$ with:
\[\delta(q,a) = \begin{cases} b & a = (1,b) \\ q & a = (0,b) \end{cases}\]
This is a reverse reset transparent Moore Machine, and so $R^{\trunc}$ can be written in terms of character-wise maps and $RBit^{\trunc}$.

Consider the action of $R^{\trunc}$ on $\cw_{pair}(\mask(w),AS_n(w))$:

\[\begin{array}{c|cccccc}
w & a & b & c & d & e & \\
\hline
\mask(w) & 0 & 0 & 0 & 0 & 0 & 1\\
AS_n(w) & id & a & ab & abc & abcd & abcde \\
R^{\trunc}(\cw_{pair}(``)) & abcde & abcde & abcde & abcde & abcde & abcde
\end{array}\]

I now have a word which can be combined with $AS_n(w)$ via the appropriate bitwise map (multiplying by the inverse on the left) to produce $RAS_n(\cw_{inverse}(w))$.

However, what I wanted was $RAS_n^{\trunc}(\cw_{inverse}(w))$. One can attain that by combining $RAS_n(\cw_{inverse}(w))$ with $\mask(w)$ bitwise to replace the last character of $RAS_n(\cw_{inverse}(w))$ with a $\#$ and then using $\unpad$ to remove it.

Letting $f$ map pairs of the form $(0,a)$ to $a$ and pairs of the form $(1,a)$ to $\#$:

\[\begin{array}{c|cccccc}
w & a & b & c & d & e & \\
\hline
\mask(w) & 0 & 0 & 0 & 0 & 0 & 1\\
RAS_n(\cw_{inverse}(w)) & (abcde)^{-1} & (bcde)^{-1} & (cde)^{-1} & (de)^{-1} & e^{-1} &  id\\
\cw_f(``) &  (abcde)^{-1} & (bcde)^{-1} & (cde)^{-1} & (de)^{-1} & e^{-1} &  \#\\
\unpad(``) & (abcde)^{-1} & (bcde)^{-1} & (cde)^{-1} & (de)^{-1} & e^{-1} & \\
\end{array}\]

This composition produces the desired output and works in general. Thus, the final form of my theorem:

\begin{thm}
Any regular function can be written as a multivariable composition of:
\begin{itemize}
\item $AS_n^{\trunc}$ for various $n$,
\item $Bit^{\trunc}$,
\item $RBit^{\trunc}$,
\item Character-wise maps,
\item $\tuplefy$,
\item $S_a$ for various $a$,
\item $\unpad$.
\end{itemize}
\end{thm}

\section{Further Research}

While of independent interest, these decomposition results have significantly streamlined proofs involving nonstandard models of the Weak Second order Theory of One Successor (WS1S). I intend to publish my findings in two papers, one presenting a complete axiomatization of WS1S, and the other presenting some results towards a classification of the nonstandard models of WS1S.

While thinking about automata as transducers instead of acceptors takes one away from the underlying logic, it takes one closer to real-world applications of automata. One is then lead to ask similar questions about functions whose graphs are recognized by B{\"{u}}chi automata, Tree automata, and Rabin automata. The determinization-harvester decomposition can be adapted to trees, but is there an analog to the Krohn-Rhodes theorem for trees in this context? What about in the case of B{\"{u}}chi automata or Rabin automata, for which there is no end of the input to start the harvester running backwards from? Can nice generators still be found?

There is still much work to be done in establishing a B{\"{u}}chi-Elgot-Trakhtenbrot theorem for graphs. There are several nice candidates for a monadic second-order logic of graphs, and some nice notions of automata operating on graphs, but no full correspondence between them. The current state of the art is Courcelle's Theorem (see, e.g. \cite{downey}), which allows one to translate questions in a graph logic to tree automata operating on a tree decomposition of the original graph, but not back. Fortunately, this is the direction of most interest to applications. But perhaps an approach which instead of trying to connect formulas and acceptors, connected describable functions and transducers, would shed light on the problem?

Finally, also note that it is not possible to break down the $AS_n$ generators much further. Of course for large enough $n$ one may decompose the symmetric group $S_n$ as a semidirect product of the alternating group $A_n$ and $S_2$, and use this decomposition to guide a slight decomposition of $AS_n$, but this doesn't gain anything. In light of the simplicity of the (large) alternating groups, it is likely no further decomposition in the function composition context is possible. One would then like a proof that, for instance, accumulators on the cyclic groups do not suffice, in a way that hopefully sheds some light on what behavior symmetric groups capture that cyclic groups cannot. Alternately, a decomposition of the accumulators on the symmetric groups in terms of accumulators on the cyclic groups would be a remarkable result.

\section*{Acknowledgment}
  \noindent The author wishes to acknowledge fruitful discussions with Scott Messick. This work would not have been possible without encouragement and guidance from my Ph.D. advisor, Anil Nerode.

\appendix
\section{List of Notation}
\begin{tabular}{ll}
$M,N$ & Moore machines\\
$R,P$ & Reverse Moore machines\\
$A$ & Finite automata\\
$\Sigma, \Gamma$ & Finite alphabets\\
$\Sigma^*$ & The set of finite words with characters from $\Sigma$\\
$Q$ & Finite set of states\\
$q_0$ & Initial state\\
$q_f$ & Final state\\
$I$ & Set of initial states\\
$F$ & Set of final states\\
$\delta$ & Transition relation $Q \times \Sigma \to Q$\\
$\delta_a(q)$ & $= \delta(q,a)$\\
$|w|$ & The length of the word $w$\\
$w[i]$ & The character of word $w$ in position $i$ (first character is $w[0]$, last is $w[|w|-1]$)\\
$w_0^{\frown}w_1$ & Concatenation\\
$S_a(w)$ & $=w^{\frown}a$\\
$\trunc$ & Remove the last character of a word\\
$\rest$ & Remove the first character of a word\\
$M(w)$ & The output of Moore machine $M$ on input $w$\\
$M^{\trunc}(w)$ & $=\trunc(M(w))$\\
$M^{\rest}(w)$ & $=\rest(M(w))$\\
$\#$ & Dummy character for padding ends of words\\
$\tuplefy$ & Combine inputs words in parallel, padding with $\#$\\
$\cw_f$ & Character-wise application of function $f$\\
$\det(A)$ & Determinization of automaton $A$\\
$\harv(A)$ & Harverster construction for automaton $A$\\
$\pair(x,y)$ &$= (x,y)$\\
$\unpad$ & Removes terminal $\#$ characters\\
$\overline{M}$ & Moore machine which keeps track of a state that Moore machine $M$ is not in\\
$\widehat{M}$ & Moore machine which keeps track of the rest of the information about the state of $M$\\
$\ors_{k,n}$ & The set of ordinal removal sequences for $n$ of length $k$\\
$\mathbf{k}$ & $=(k_0,\ldots,k_{i-1})$\\
$\remove(L,\mathbf{k})$ & Applies ordinal removal sequence $\mathbf{k}$ to ordered collection $L$\\
$\vec{M}$ & A reset Moore Machine\\
$\tilde{M}$ & A permutation Moore Machine\\
$AS_n$ & The accumulator on $S_n$ automaton\\
$Bit$ & The bit storage automaton\\
$RAS_n$ & The reverse accumulator on $S_n$\\
$RBit$ & The reverse bit storage automaton\\
$\rev(w)$ & The reverse of word $w$
\end{tabular}


\begin{thebibliography}{Kos97}

\bibitem{sql}
Benedikt, M., Libkin, L., Schwentick, T., \& Segoufin, L. (2003). Definable relations and first-order query languages over strings. \emph{Journal of the ACM}, 50(5), 694-751.

\bibitem{downey}
Downey, R. G., \& Fellows, M. R. (1999). \emph{Parameterized complexity}. New York: Springer.

\bibitem{nagy05}
Egri-Nagy, A., \& Nehaniv, C. (2005). 
Algebraic Hierarchical Decomposition of Finite State Automata: Comparison of Implementations for Krohn-Rhodes Theory. \emph{Implementation and Application of Automata Lecture Notes in Computer Science},
315-316.

\bibitem{nagy06}
Egri-Nagy, A., \& Nehaniv, C. (2006). 
Making Sense of the Sensory Data – Coordinate Systems by Hierarchical Decomposition. 
\emph{Lecture Notes in Computer Science Knowledge-Based Intelligent Information and Engineering Systems},
333-340.

\bibitem{eilenberg}
Eilenberg, S., \& Tilson, B. (1976). \emph{Automata, Languages, and Machines: Volume B.} New York: Academic Press.

\bibitem{ginzburg}
Ginzburg, A. (1968). 
\emph{Algebraic theory of automata}. 
New York: Academic Press.

\bibitem{krohn}
Krohn, K., \& Rhodes, J. (1965). Algebraic Theory of Machines. I. Prime Decomposition Theorem for Finite Semigroups and Machines \emph{Transactions of the American Mathematical Society}, 450-450.

\end{thebibliography}
\end{document}